\documentclass[10pt,aps,prl,twocolumn,nofootinbib,eqsecnum]{revtex4-2}
\usepackage[utf8]{inputenc}
\usepackage{amssymb,amsmath,amsthm,mathtools}
\usepackage[english]{babel}

\usepackage[breaklinks]{hyperref}

\newtheorem{theorem}{Theorem}[section]

\newtheorem{lemma}[theorem]{Lemma}

\makeatletter
\renewcommand{\c@secnumdepth}{0}
\makeatother

\newcommand{\tr}{{\text{tr}}}
\newcommand{\B}{\textbf}

\begin{document}

\title{Characterizing quantum state-space with a single quantum measurement}
\author{Matthew B. Weiss}
\affiliation{University of Massachusetts, Boston: QBism Group}
\date{\today}

\begin{abstract}
Can the state-space of $d$-dimensional quantum theory be derived from studying the behavior of a single ``reference" measuring device? The answer is yes, if the measuring device corresponds to a complex-projective 3-design. In this privileged case, not only does each quantum state correspond to a probability-distribution over the outcomes of a single measurement, but also the probability-distributions which correspond to quantum states can be elegantly characterized as those which respect a generalized uncertainty principle. The latter takes the form of a lower-bound on the variance of a natural class of observables as measured by the reference. We give simple equations which pure-state probability distributions must satisfy, and contextualize these results by showing how 3-designs allow the structure-coefficients of the Jordan algebra of observables to be extracted from the probabilities which characterize the reference measurement itself. This lends credence to the view that quantum theory ought to be primarily understood as a set of normative constraints on probability assignments which reflect nature's lack of hidden variables, and further cements the significance of 3-designs in quantum information science.
\end{abstract}

\maketitle

\section{Introduction}

In 1927, Niels Bohr introduced the notion of \emph{complementarity}, that not all aspects of a physical system may be simultaneously definite, as the distinctive feature of the new quantum mechanics \cite{Bohr1928-vp, baggott2011quantum}. For example, the Heisenberg uncertainty principle, $\sigma_x \sigma_p\ge \frac{1}{2}\hbar$, tells us that if we experiment upon an ensemble of identically prepared particles, then a small variance in the measured position of the particles implies a large variance in their momentum, and vice versa. In particular, regardless of the choice of ensemble, the variance of the two quantities cannot be made arbitrarily small together while remaining consistent with quantum theory. The uncertainty principle may be generalized e.g.,\ to arbitrary pairs of observables, and even to collections of observables \cite{hou2016uncertaintyrelationsmultiobservables}. If Bohr was right that complementarity is the defining feature of quantum theory, then it ought to be possible to characterize ``quantum states'' as nothing other than probability-assignments which satisfy appropriate uncertainty relations for all possible observables.

From this point of view, what is essential is not the traditional Hilbert space formalism, but instead the constraints quantum theory urges on probability-assignments: indeed, the former can be seen as a convenient mathematical technique for imposing those very constraints. This is the central contention of QBism \cite{fuchs2019qbismquantumtheoryheros}, a subjective Bayesian interpretation of quantum mechanics, which holds that quantum theory should not be viewed as a description of physical reality but rather as a set of consistency-constraints on probability-assignments motivated by nature's lack of hidden variables. In particular, QBists side with Schr\"odinger in viewing the quantum state as nothing more than a ``catalogue of expectations'' \cite{cat}. Indeed, while Bohr and Heisenberg focused on relations between observables like position and momentum, contemporary quantum information theory contemplates a more general class of measurements, so-called \emph{informationally-complete} (IC) measurements \cite{Caves_2002}. Remarkably, assigning appropriate probabilities to the outcomes of a single IC measurement is equivalent to assigning a quantum state. Thus one may take Schr\"odinger's  ``catalogue of expectations'' a step further, and identify quantum states with probability-distributions directly.

 The caveat is that while all quantum states correspond to probability-distributions, not all probability-distributions correspond to quantum states \cite{appleby2011propertiesqbiststatespaces}. On the Hilbert space side, such invalid distributions correspond to self-adjoint matrices which are not positive-semidefinite, and thus cannot be regarded as density matrices. A picture thereby emerges of quantum state-space as a privileged \emph{subset} of the probability-simplex, corresponding to just those probability-distributions which map back to valid states. The shape of the subset depends on the choice of informationally-complete measurement used as a ``reference,'' and can be derived by calculating certain quantities associated with the Hilbert space representation of the measurement. From a foundational perspective, however, one might wonder whether it is possible to characterize the shape of such a privileged subset without reference to Hilbert space as such. Related questions have recently been asked in the context of GPT tomography \cite{PRXQuantum.2.020302, PhysRevA.105.032204}, which provides a theory-agnostic approach to reconstructing state and effect spaces from the experimental data collected from a wide range of preparation and measurement procedures. Certainly, a state is invalid if it implies a negative probability for some effect. From a QBist point of view, as we shall see, one may rewrite the Born rule entirely in terms of probabilities with respect to a reference measurement. One may thus rule out a probability-distribution if the Born rule, formulated in terms of probabilities, yields a negative number for some measurement outcome. In this way, the shape of state-space would emerge out of the demand for probabilistic consistency.
  
 Nevertheless, just as a single IC reference measurement may characterize the preparation of a quantum system entirely in terms of reference probabilities, one might wonder whether a single specially chosen reference device could perform ``state-space tomography'' by appealing only to probabilities assigned to that single device. Probabilistic consistency with this one device would imply consistency with every possible measurement, and provide an important proof of principle for the QBist approach to quantum mechanics. In this work, we demonstrate that this can in fact be done. Taking one's reference measurement to be a \emph{complex-projective 3-design} allows the set of valid probability-distributions---that is, the shape of quantum state-space---to be elegantly characterized by an uncertainty principle\footnote{We note that whereas the usual uncertainty principle relates probability-distributions with respect to different measurements, our generalization is formulated with respect to a single measurement.}, precisely in the spirit of Bohr. Crucially, the only building blocks needed to formulate this principle are reference device probabilities.

  In particular, we will show that probability-assignments to the outcomes of the reference device cannot be too sharp in a prescribed way: they must satisfy a lower-bound on the variance with respect to any of a natural class of observables. This is so even as those same probability-assignments may imply a sharp distribution for the outcomes of some alternative measurement via the Born rule. This is a dramatic reversal of the situation classically, where any alternative measurement may be regarded as a coarse-graining of a reference measurement e.g., of the positions and momenta of a set of particles. In the classical case, on the one hand, certainty is achievable for the reference; on the other hand, one cannot achieve more certainty about alternative measurements than about the reference. Thus our result underscores in a novel and perspicacious way the degree to which quantum mechanics resists hidden variable interpretations. Moreover, the notion of defining the set of valid probability-distributions in terms of an uncertainty principle is a promising approach to constructing generalizations of quantum theory in the spirit of  \emph{generalized probabilistic theories} \cite{M_ller_2021, barnum2013postclassicalprobabilitytheory} or in the QBist literature, the \emph{qplex} research program \cite{Appleby_2017}. Independently motivating the features of a 3-design representation, and demonstrating how a Hilbert space representation arises from them (rather than the other way around), would represent a significant advance in the ongoing quest to derive the quantum formalism from satisfying quantum information-theoretic principles, and at the same time likely lead to a novel characterization of 3-designs themselves. Indeed, 3-designs are of great contemporary interest due to their special role in the theory of classical shadow estimation, where employing a 3-design allows one for example to achieve constant sample complexity in fidelity estimation independent of system size \cite{Huang_2020,mao2024magicquditshadowestimation,Kliesch_2021, chen2024nonstabilizernessenhancesthriftyshadow}. We hope the present work places these developments in a broader context.
  
We begin by reviewing the probability-first formalism for quantum mechanics furnished by an \emph{informationally-complete reference device}, and then lay out the basic features of complex projective $3$-designs. These mathematical preliminaries aside, we derive scalar constraints on probability-assignments corresponding to pure-states: pure-state probability-assignments turn out to live in the intersection of 2-norm and 3-norm spheres of specified radii restricted to a natural subspace. Operationally, these constraints can be interpreted as bounds on the agreement between several copies of the reference device when fed identically prepared systems; conceptually, they can be understood as entropic uncertainty principles \cite{PhysRevLett.60.1103}. We then establish a vector constraint on pure-probability assignments, which draws our attention to a particular 3-index tensor built out of the probabilities which characterize the reference device itself. To handle the case of mixed states, we report the key insight that for a 3-design measurement, one can relate the variance of an observable as measured directly to the variance of the observable as measured by the reference. This is what allows quantum state-space to be characterized in its entirety by a single uncertainty principle, which lower-bounds the variance of any of a natural class of observables. In closing, we observe that the reason 3-designs play such a privileged role is that for just these measurements, the structure-coefficients for the Jordan algebra of observables \cite{Faraut1994} can be expressed solely in terms of the probabilities which characterize the reference device.  The essence of this observation was already made in \cite{obst2024wignerstheoremstabilizerstates, Heinrich_2019}, but its significance for probabilistic representations of quantum mechanics was left unexplored.

\section{The reference device formalism} 

In quantum mechanics, the most general form of a measurement with a finite number of outcomes consists in a set $\{E_i\}_{i=1}^n$ of positive-semidefinite matrices called \emph{effects}, acting on a Hilbert space $\mathcal{H}_d$ satisfying $\sum_{i} E_i = I$. If the effects span the $d^2$-dimensional operator space, we call the measurement \emph{informationally-complete}:  the probabilities $P(E_i|\rho)=\tr(E_i\rho)$ fully characterize the density matrix $\rho$ representing a quantum state, which must itself be positive-semidefinite with $\tr(\rho)=1$.  We call a measure-and-reprepare device which performs an informationally-complete measurement and conditional on the outcome, prepares one of a set $\{\sigma_i\}_{i=1}^n$ of informationally-complete reference states, a \emph{reference device}. In what follows, we will always assume that our reference device prepares a reference state $\sigma_i = \tr(E_i)^{-1}E_i$ proportional to the reference effect representing the outcome.

The minimal number of effects in an IC-measurement is $d^2$: moreover, at best an IC-measurement may furnish a linearly-independent, but not orthonormal basis \cite{e16031484, cuffaro2024quantumstatesmaximalmagic, debrota2020varietiesminimaltomographicallycomplete}, and more generally an informationally-\emph{over}complete set. We must therefore take up the matter of its dual representation---with a probabilistic  twist. Let $|\sigma_i) =\text{vec}(\sigma_i)= (I \otimes \sigma_i)\sum_i |i,i\rangle$ be the vectorization of a reference state $\sigma_i$, and similarly let $(E_i|=\sum_i\langle i,i|(I \otimes E_i)$ be the vectorization of a reference effect. In general,  $(A|B)=\tr(A^\dagger B)$, and so arranging $(E_i|$ into the rows of a matrix $\B{E}$, and $|\sigma_i)$ into the columns of a matrix $\textbf{S}$, we can write the conditional-probability matrix with elements $P(E_i|\sigma_j)$, for the probability of a reference outcome given a reference preparation, $P\equiv \B{ES}$. By informational-completeness, these probabilities fully characterize the reference device itself, and they will play a fundamental role in the sequel.

We call a \emph{Born matrix} any matrix $\Phi$ which satifies  $P\Phi P = P$, the defining equation of 1-inverse of $P$ \cite{Ben-Israel2003-ym}. It follows from informational-completeness that $P\Phi P = P \Longleftrightarrow \B{S}\Phi \B{E}=I$ \cite{weiss2024depolarizingreferencedevicesgeneralized}, which provides a resolution of the identity, and thus a dual representation $|\rho) = \B{S}\Phi \B{E}|\rho)$ or,
\begin{align}
\rho=\sum_{ij}\Phi_{ij}P(E_j|\rho)\sigma_i.
\end{align}
If the measurement operators (and states) are linearly independent, then $P$ will be invertible, and thus $\Phi = P^{-1}$. Otherwise, there will be a variety of choices for the Born matrix\footnote{The $1$-inverses of a matrix $P=U\Sigma V^\dagger$ may all be calculated from its singular value decomposition via 
\begin{align}
	\Phi = V \begin{pmatrix}\sigma^{-1} & A \\ B & C \end{pmatrix}U^\dagger\nonumber,
\end{align}
where $A, B, C$ are completely arbitrary matrices, and $\sigma$ is the diagonal matrix of nonzero singular values. A typical example is the Moore-Penrose pseudo-inverse, for which $A=B=C=0$ \cite{Ben-Israel2003-ym}.}, and different assignments of probabilities will lead to the same ascription of a density matrix. 

Consider now some alternative measurement $\{A_i\}_{i=1}^m$. We can write the Born rule probability $P(A_i|\rho)=\tr(A_i\rho)$ in terms of reference probabilities,
\begin{align}
P(A_i|\rho) &= \tr(A_i\rho)=(A_i|\B{S}\Phi \B{E}|\rho)\\
&=\sum_{jk}P(A_i|E_j)\Phi_{jk}P(E_k|\rho)\nonumber,
\end{align}
where we have taken $P(A_i|E_j)\equiv P(A_i|\sigma_j)$, using the outcomes of the reference device as a proxy for its state preparations. One may appreciate that the Born rule has become a simple \emph{deformation} \cite{Ferrie_2011} of the law of total probability $P(A_i) = \sum_j P(A_i|E_j)P(E_j)$.

In particular, taking $\B{A}=\B{E}$, we have that $P(E_i|\rho)=\sum_{jk}P(E_i|\sigma_j)\Phi_{jk}P(E_k|\rho)$. This is a fundamental consistency-criterion in an overcomplete probability representation. It follows from the defining equation of a 1-inverse that $P\Phi$ is a projector, and in what follows we may take $\Phi$ to be invertible, so that $P\Phi$ projects onto $\text{col}(P)$: recall the column-space of $P$ (or equivalently, its range) is the span of its columns. Now $P=\B{ES}$ is in fact a full-rank factorization of $P$, and hence the columns of $\B{E}$ form a basis for $\text{col}(P)$ \cite{full_rank_factorization}. Thus any vector with components $x_i =\tr(E_iX)$, that is, $x=\B{E}|X)$, must lie in $\text{col}(P)$. Conversely, if $x\in \text{col}(P)$, there must exist a vector $\tilde{x}$ such that $x_i= \sum_j P(E_i|\sigma_j)\tilde{x}_j=\tr\left(E_i \sum_j \tilde{x}_j\sigma_j\right)=\tr(E_i \tilde{X})$ for some operator $\tilde{X}$.

\section{Making designs}

It follows from the representation theory of the unitary and symmetric groups \cite{harrow2013churchsymmetricsubspace, bacon2005quantumschurtransformi} that the \emph{$t$-th moment of quantum state-space}, that is, the $t$-th tensor power of a pure-state averaged over all pure-states is 
\begin{align}
\int |\psi\rangle\langle \psi|^{\otimes t} d\psi = \binom{d+t-1}{t}^{-1}\Pi_{\text{sym}^t},
\end{align}
where $d\psi$ denotes the Haar measure on pure-states, and $\Pi_{\text{sym}^t}$ is the projector onto the permutation-symmetric subspace on $t$ tensor-factors \cite{Waldron2018}. Note $\tr(\Pi_{\text{sym}^t})=\binom{d+t-1}{t}$ is just the dimension of that subspace. As $\Pi_{\text{sym}^t}$ can be expressed as a sum over all permutation operators, we have in fact
\begin{align}
	\int |\psi\rangle\langle \psi|^{\otimes t} d\psi &=\binom{d+t-1}{t}^{-1}\frac{1}{t!}\sum_{\pi\in S_t}T_\pi
\end{align}
where $T_{\pi} = \sum_{a,b,c, \dots}|\pi(a), \pi(b), \pi(c), \dots\rangle\langle a, b, c, \dots|$.

A \emph{quantum  state $t$-design}, also called a \emph{complex-projective t-design} \cite{Waldron2018}, is an ensemble of pure-states $\{p_i, |\psi_i\rangle\}_{i=1}^n$ which satisfy
\begin{align}
\sum_{i=1}^n p_i|\psi_i\rangle\langle \psi_i|^{\otimes t} =\int |\psi\rangle\langle \psi|^{\otimes t} d\psi :
\end{align}
the average over the design-ensemble mimics the average over all pure-states up to the $t$-th moment. We will call a design \emph{unbiased} if $\forall i: p_i=\frac{1}{n}$. Here 
\begin{align}
n \ge \binom{d-1+\lfloor t/2 \rfloor}{\lfloor t/2 \rfloor }	\binom{d-1 +\lceil t/2 \rceil}{\lceil t/2 \rceil} \nonumber.
\end{align}
We note that a $t$-design of any order always exists for sufficiently large $n$ \cite{SEYMOUR1984213}, and a $t$-design is also a $(t-1)$-design. For $t=1$, we have $\sum_i p_i|\psi_i\rangle\langle \psi_i| =\frac{1}{d}I$ which shows that a $1$-design furnishes a set of rank-1 projectors which, rescaled, sum to the identity: thus a 1-design is a quantum measurement. For a 2-design 
\begin{align}
\label{swappy}
\sum_i p_i |\psi_i\rangle\langle \psi_i|^{\otimes 2}=\frac{1}{d(d+1)}(I \otimes I + \mathcal{S}),
\end{align}
where $\mathcal{S}$ is the swap operator: two typical examples are symmetric informationally-complete (SIC) states and states corresponding to a complete set of mutually unbiased bases (MUBs) \cite{bengtsson2017discretestructuresfinitehilbert}. In fact, it  follows from Eq.\ \ref{swappy} that for any unbiased 2-design reference device, we can take $\Phi=(d+1)I - \frac{d}{n}J$, where $J$ is the Hadamard identity (Appendix \ref{2-design-Phi}).  The Born rule then takes the profoundly elegant form
\begin{align}
P(A|\rho) &=\sum_i P(A|E_i)\left[(d+1)P(E_i|\rho)-\frac{d}{n}\right],
\end{align}
 which was given an independent motivation in the case that $n=d^2$ in \cite{DeBrota_2020}, and is related to the fact that 2-designs are optimal for linear quantum state tomography \cite{Scott_2006, Zhu_2011}. The terrain of 3-designs is an area of active investigation: several examples were presented in \cite{S_omczy_ski_2020}, and more general constructions were given in \cite{Gross_2021, zhu2024momentsquditcliffordorbits}. In particular, 3-designs have been studied for their use in the theory of classical shadows \cite{Huang_2020,mao2024magicquditshadowestimation,Kliesch_2021} since they lead to well-controlled variance in the estimation of expectation values. In $d=2$, the complete set of MUBs forms a $3$-design, as does the set of $n$-qubit stabilizer states \cite{Zhu_2017}.  When $t=3$, $n\ge \frac{1}{2}d^2(d+1)$, a bound which is not tight \cite{Waldron2018}.

\section{The shape of quantum state-space}

\subsection{Bounding agreement}
\label{Agreement-probabilities}

We begin by providing a characterization of pure-state probability-distributions in terms of a set of entropic uncertainty principles. They can be motivated by considered the following scenario. Suppose we prepare states $\rho_1, \dots, \rho_t$, and send each into its own copy of a reference device $\{E_i, \sigma_i\}_{i=1}^n$. What is the probability that all $t$ reference devices give the same outcome? In other words, we are interested in the agreement-probability
\begin{align}
&P(\text{agree}|\rho_1, \dots, \rho_t)\\
&=\sum_{i=1}^n \prod_{j=1}^t P(E_i|\rho_j)= \tr\left(\sum_{i=1}^n E_i^{\otimes t} \otimes_{j=1}^t \rho_j\right)\nonumber.
\end{align}
Assuming the reference device is unbiased and that the reference states are proportional to effects $(E_i = \frac{d}{n}\sigma_i)$, we have $P(\text{agree}|\rho_1, \dots, \rho_t)= \frac{d^t}{n^{t-1}}\tr\left(\frac{1}{n}\sum_{i=1}^n \sigma_i^{\otimes t} \otimes_{j=1}^t \rho_j\right)$. If we further assume that the reference device forms a $t$-design,  $\frac{1}{n}\sum_i\sigma_i^{\otimes t} = \binom{d+t-1}{t}^{-1}\Pi_{\text{sym}^t}$, and so the agreement-probability can be written
\begin{align}
&P(\text{agree}|\rho_1, \dots, \rho_t)\\
&=\frac{d^t}{n^{t-1}}\binom{d+t-1}{t}^{-1}\frac{1}{t!}\sum_{\pi\in S_t}\tr(T_\pi\otimes_{j=1}^t \rho_j)\nonumber.
\end{align}
To evaluate expressions like this, it suffices to consider traces with cyclic permutations. For $t=2$, the swap operator $\mathcal{S}=\sum_{ab}|b,a\rangle\langle a, b|$ yields
\begin{align}
&\tr\left(	(X\otimes Y)\sum_{ab}|b,a\rangle\langle a, b|\right)\\
&=\sum_{ab}\langle a|X|b\rangle \langle b|Y|a\rangle=\tr(XY)\nonumber,
\end{align}
while similarly, for $t=3$, a cyclic permutation  of three elements delivers
\begin{align}
&\tr\left(	(X \otimes Y \otimes Z)\sum_{abc}|b,c,a\rangle\langle a, b,c|\right)\\
&=\sum_{abc} \langle a|X|b\rangle\langle b|Y|c\rangle\langle c|Z|a\rangle=\tr(XYZ)\nonumber.
\end{align}
In light of this, we have for the order-2 agreement-probability,
\begin{align}
P(\text{agree}|\rho_1, \rho_2)
&=\frac{1}{d+1}\left(\frac{d}{n}\right)\Big[\tr(\rho_1)\tr(\rho_2)+\tr(\rho_1\rho_2)\Big].
\end{align}
Now $\tr(\rho_1\rho_2)\leq \sqrt{\tr(\rho_1^2)\tr(\rho_2^2)}$, while the purity satisfies $\tr(\rho^2)\leq 1$. Thus when $\rho_1=\rho_2=\rho$ pure, we saturate the upper-bound of this quantity. For the lower-bound, we note $\tr(\rho_1\rho_2)\ge0$ with equality if and only if $\rho_1, \rho_2$ are orthogonal, which leads to bounds
\begin{align}
\left(\frac{d}{n}\right)\frac{1}{d+1}\leq\sum_i P(E_i|\rho_1)P(E_i|\rho_2)\leq \left(\frac{d}{n}\right)\frac{2}{d+1}.
\end{align}
Similarly, for $t=3$, we find
\begin{small}
\begin{align}
&P(\text{agree}|\rho_1, \rho_2, \rho_3)\\
 &= \frac{1}{(d+1)(d+2)}\left(\frac{d}{n}\right)^2\Big[\tr(\rho_1)\tr(\rho_2)\tr(\rho_3) +\tr(\rho_1)\tr(\rho_2\rho_3)\nonumber\\
&+\tr(\rho_2)\tr(\rho_1\rho_3)+\tr(\rho_3)\tr(\rho_1\rho_2)+\tr(\rho_1\rho_2\rho_3)+\tr(\rho_1\rho_3\rho_2)\Big]\nonumber,
\end{align}
\end{small}
\!\!which is maximized when $\rho_1=\rho_2=\rho_3=\rho$ pure, so that $\tr(\rho^2)=\tr(\rho^3)=1$, and minimized when $\rho_1, \rho_2, \rho_3$ are mutually orthogonal, delivering bounds
\begin{align}
\left(\frac{d}{n}\right)^2\frac{1}{(d+1)(d+2)} &\leq \sum_i P(E_i|\rho_1)P(E_i|\rho_2)P(E_i|\rho_3)\nonumber\\
&\leq \left(\frac{d}{n}\right)^2\frac{6}{(d+1)(d+2)}.
\end{align}
Since in both cases, the upper-bounds are saturated by identical pure-states, we conclude that pure-state probability assignments with respect to a 3-design lie in the nonnegative orthant, in the intersection of three kinds of spheres. From $\sum_i P(E_i|\rho)$, they live on a 1-norm sphere of radius 1; from $\sum_i P(E_i|\rho)^2$, they live on a 2-norm sphere of radius $\sqrt{\left(\frac{d}{n}\right)\frac{2}{d+1}}$; and from $\sum_i P(E_i|\rho)^3$, they live on a 3-norm sphere of radius $\sqrt[3]{\left(\frac{d}{n}\right)^2\frac{6}{(d+1)(d+2)}}$. Finally, since we derived all probabilities from the trace of a reference effect on a state, our probability-distributions live in $\text{col}(P)$.

We could continue on, contemplating the $t$-fold agreement-probability for an unbiased $t$-design reference device for any $t$. Since a pure-state satisfies $\forall t: \tr(\rho^t)=1$, we'd find that pure probability vectors live on $t$-norm spheres with fixed radii determined by the agreement-probability
\begin{align}
&P(\text{agree}|\rho^{\otimes t}) \\
&= \sum_i P(E_i|\rho)^t =\frac{d^t}{n^{t-1}}\binom{d+t-1}{t}^{-1}\frac{1}{t!}\sum_{\pi\in S_t}\tr(T_\pi\rho^{\otimes t})\nonumber\\
&= \frac{d^t}{n^{t-1}}\binom{d+t-1}{t}^{-1},
\end{align}
where the last follows from the fact that since $\forall t: \tr(\rho^t)=1$, all of the $t!$ terms in the sum will be 1.

The following lemma \cite{Jones_2005, fuchs2015strugglesblockuniverse}, however, assures us that a $3$-design is all we need.
\begin{lemma}
	A Hermitian operator $A$ is a rank-1 projector if and only if $\tr(A^2) = \tr(A^3)=1$. 
\end{lemma}
\begin{proof}
Let $\{\lambda_i\}$ be the eigenvalues of $A$. $\tr(A^2) = \tr(A^3)=1$  means that $\sum_i \lambda_i^2 = \sum_i \lambda^3 = 1$.
On the one hand, $\sum_i \lambda_i^2 = 1$ implies that $\forall i: -1 \leq \lambda_i \leq 1$. On the other hand, $\sum_i \lambda_i^3 \leq \sum_i \lambda_i^2$ with equality if and only if $\forall i: \lambda_i \in \{0, 1\}$. But since the whole sum must be 1, we must have exactly one $\lambda_i=1$ and the rest  0. Thus $A$ is a rank-1 projector, or equivalently a pure-state $\rho$.
\end{proof}
\noindent In light of this, as long $t\ge 3$, we can fully characterize the pure-states of quantum theory with respect to an unbiased $t$-design reference device by $\forall i: P(E_i|\rho)\ge0$ and 
\begin{align}
\sum_i P(E_i|\rho)&= 1	\\
\sum_i P(E_i|\rho)^2&= \left(\frac{d}{n}\right)\frac{2}{d+1}	\\
\sum_i P(E_i|\rho)^3 &= \left(\frac{d}{n}\right)^2\frac{6}{(d+1)(d+2)},
\end{align}
along with the consistency-condition required by overcompleteness $P(E_i|\rho)=\sum_{jk}P(E_i|\sigma_j)\Phi_{jk}P(E_k|\rho)$, that is, the probability-distribution lives in $\text{col}(P)$. In Appendix \ref{scalar-overcompleteness}, we show more specifically that a probability-distribution corresponds to a pure state if and only if its projection into $\text{col}(P)$ satisfies the quadratic and cubic constraints.

 Thus in a sense we have already achieved our goal: quantum state-space is the convex hull of all pure-state probability-assignments. The bounds on agreement-probabilities, from which the constraints on pure-states were derived, already suggest a kind of uncertainty principle: after all, classically, there is nothing in principle preventing the agreement-probability to be 1 when identical preparations are fed into each reference device or 0 for perfectly distinguishable preparations. In fact, the upper bounds may be understood as implying a set of \emph{entropic} uncertainty principles. Defining the order-$t$ R\'enyi entropy of a probability-distribution $\{P(E_i|\rho)\}$ as
\begin{align}
H_t\Big(\{P(E_i|\rho)\}\Big) &= \frac{1}{1-t}\log\left(\sum_i P(E_i|\rho)^t\right)	
\end{align}
for $0< t <\infty$, it is clear from the above discussion that pure-state probability-distributions achieve the lower-bound on the R\'enyi entropies of order $t=2,3$ over all states, and this, along with the restriction to $\text{col}(P)$, is enough to characterize them completely.

\subsection{The contour of idempotents}
\label{idempotents}

As an alternative approach, we can derive a single equation picking out pure-state probability-assignments by appealing to the fact that, for a normalized state, $\rho=\rho^2$ if and only if $\rho$ is pure. Substituting the resolution of the identity $\rho=\sum_{ij}\Phi_{ij}P(E_j|\rho)\sigma_i$ into $P(E_i|\rho^2)=\tr(E_i\rho^2)$, we find that
\begin{align}
\label{rho_rhosq}
&P(E_i|\rho)  \\
&= 	\sum_{lm} P(E_l|\rho)P(E_m|\rho)\sum_{jm}\Phi_{jl}\Phi_{km}\Re\big[\tr(E_i\sigma_j\sigma_k)\big]\nonumber.
\end{align}
We need only consider the real-part since  $\tr(E_i \sigma_k\sigma_j)=\tr\big( E_i (\sigma_j \sigma_k)^\dagger\big)=\tr\big( E_i (\sigma_j^* \sigma_k^*)^T\big)=\tr\big( E_i^T \sigma_j^* \sigma_k^*\big)=\tr\big( E_i^* \sigma_j^* \sigma_k^*\big)=\tr(E_i\sigma_j \sigma_k)^*$ as $E_i, \sigma_j, \sigma_k$ are all positive-semidefinite and so every term in Eq.\ \ref{rho_rhosq} is added to its complex conjugate.

 Let $\mathcal{M}_t= \int |\psi\rangle\langle\psi|^{\otimes t}d\psi$ be the $t$-th moment of quantum state-space. This is itself a valid state, and so we can consider its probability-distribution $P(E_i, E_j, E_k, \dots|\mathcal{M}_t)$ with respect to $t$ copies of the reference measurement. If we assume an unbiased set of effects, by the same argument as in Section \ref{Agreement-probabilities}, it follows that
 \begin{small}
\begin{align}
&P(E_i|\mathcal{M}_1) = \frac{1}{n}	\\
&P(E_i, E_j|\mathcal{M}_2) =\frac{1}{d+1}\left(\frac{1}{n}\right)\left[\frac{d}{n} + P(E_i|\sigma_j)\right]\label{m2}\\
\label{any_triples}&P(E_i, E_j, E_k|\mathcal{M}_3)=\frac{1}{(d+1)(d+2)}\left(\frac{d}{n^2}\right)\times \\
&\Bigg[ \frac{d}{n}+  P(E_j|\sigma_k)  +P(E_i|\sigma_j) +P(E_i|\sigma_k) + 2\Re \big[\tr(E_i \sigma_j \sigma_k)\big]\Bigg].\nonumber	
\end{align}
\end{small}
\!\!If we further assume that the reference states form a 3-design then $\mathcal{M}_3=\frac{1}{n}\sum_i \sigma_i^{\otimes 3}$ and so $P(E_i, E_j, E_k|\mathcal{M}_3)= \frac{1}{n}\sum_{m}P(E_i|\sigma_m)P(E_j|\sigma_m)P(E_k|\sigma_m)$. Equating these two expressions allows us to calculate $\Re \big[\tr(E_i \sigma_j \sigma_k)\big]$ directly from the conditional-probability matrix $P(E_i|\sigma_j)$ which characterizes the reference device itself (Appendix \ref{big-real-part}). Then Eq.\ \ref{rho_rhosq}, which expresses $\rho=\rho^2$ in terms of probability-assignments,  simplifies to
\begin{align}
\label{quadratic}
&P(E_i|\rho)\\ &=\frac{1}{2}\Bigg[\frac{1}{2}(d+1)(d+2)\left(\frac{n}{d}\right)\sum_m P(E_i|\sigma_m)P(E_m|\rho)^2-\frac{d}{n}\Bigg]\nonumber,
\end{align}
which depends only upon $P(E_i|\sigma_j)$ and $P(E_i|\rho)$. A probability-distribution satisfying Eq.\ \ref{quadratic} is clearly in $\text{col}(P)$, and it is straightforward to check that the Eq.\ \ref{quadratic} implies the scalar constraints proved in the previous section.

Finally, we note that Eq.\! \ref{any_triples} implies that access to preparations of $\mathcal{M}_3$ allow one to extract $\Re[\tr(E_i\sigma_j\sigma_k)]$ from the joint probability distribution $P(E_i, E_j, E_k|\mathcal{M}_3)$ for \emph{any} measurement $\{E_i\}$. One may compare this method to the procedure described in \cite{oszmaniec2021measuringrelationalinformationquantum}, which exploits the backaction on a control qubit after a controlled cyclic permutation to extract the real-part of a trace of an $n$-product of states from an expectation value. Having estimated $\Re[\tr(E_i\sigma_j\sigma_k)]$ for any informationally-complete measurement, we may characterize pure-state probability-distributions with respect to that measurement according to Eq.\! \ref{rho_rhosq}. What makes a 3-design distinctive is that $\Re[\tr(E_i\sigma_j\sigma_k)]$ can be calculated from $P(E_i|\sigma_j)$ alone, and so pure-state probability-distributions can be characterized by appealing to probabilities assigned to one single device.

\subsection{A variance-based uncertainty principle}

We now give a condition for the validity of \emph{any} probability-assignment, pure or mixed. For a 3-design measurement, the variance of an observable as measured by a standard von Neumann measurement can be directly related to the variance of the same observable as estimated by the reference device, modulo overcompleteness. From this consideration, we can characterize the validity of any probability-distribution in terms of a lower-bound on the variance of all observables of a natural class.

We begin by noting that $\rho$ is positive-semidefinite if and only if its second-moment with respect to all Hermitian observables $X$ is nonnegative:
\begin{align}
\forall X: \tr(X^2\rho)\ge0 \Longleftrightarrow \rho \ge0.	
\end{align}
This follows from the fact that $X^2\ge0$ by construction, and the fact that the cone of positive-semidefinite matrices is self-dual \cite{Faraut1994}. More simply, one can observe that $\rho\ge 0$ is equivalent to $\forall \psi: \langle\psi|\rho|\psi\rangle=\tr(|\psi\rangle\langle \psi|\rho)\ge0$, and any $X^2$ can be decomposed into a sum of rank-1 projectors weighted by positive numbers via the spectral decomposition. 

Substituting  $X=\sum x_i E_i$ and $\rho=\sum_{ij}\Phi_{ij}P(E_j|\rho)\sigma_i$, we find that
\begin{align}
\label{vn_var}
&\forall X: \tr(X^2\rho) =\\
&\left(\frac{d}{n}\right)\sum_{ijkl} x_ix_j\Re[\tr(E_i\sigma_j\sigma_k)]\Phi_{kl}P(E_l|\rho) \ge 0.	\nonumber
\end{align}
As shown in  Appendix \ref{variance-overcompleteness}, if we assume that $x \in \text{col}(P)$, and exploit the expression for $\Re[\tr(E_i\sigma_j\sigma_k)]$ in terms of $P(E_i|\sigma_j)$, $\tr(X^2\rho)$ reduces to
\begin{align}
\label{simplified_var}
&\tr(X^2\rho)\\
&=\frac{1}{2}\left(\frac{d+2}{d+1}\right)\Bigg[\langle X^2\rangle_\rho-\frac{d}{d+2}\Big(\langle X^2\rangle_\mu-2\langle X\rangle_\mu\langle X\rangle_\rho\Big)\Bigg],\nonumber
\end{align}
where e.g. $\langle X^2\rangle_\rho = \sum_i x_i^2 P(E_i|\rho)$ and $\forall i: P(E_i|\mu)=\frac{1}{n}$ are the probabilities for the maximally mixed state. The restriction that $x \in \text{col}(P)$ is natural since already any Hermitian observable $X$ can be expressed with respect to the reference device in this form. Remarkably, this expression relates the second-moment of $X$ with respect to a standard von Neumann measurement, whose outcomes are eigenvalues of $X$, to the second-moment of $X$ \emph{with respect to the reference device}, where now the $x_i$'s are interpreted as numerical values assigned to reference measurement outcomes. Only for a 3-design is such a simple relationship possible, and this is what will allow us to characterize valid probability-distributions in terms of a lower-bound on the variance with respect to the reference device.

The restriction  $x \in\text{col}(P)$ (for real $x$) is equivalent to the assumption that $x_i=\tr(E_i\tilde{X})$ for some Hermitian $\tilde{X}$. By the 2-design property, if $X=\sum_i x_i E_i$ for $x_i=\tr(E_i\tilde{X})$, then $\tilde{X}=\left(\frac{n}{d}\right)\Big[(d+1)X - \tr(X)I\Big]$. Consequently, if the RHS of Eq.\ \ref{simplified_var} is nonnegative for some real $x \in \text{col}(P)$, then $\tr(X^2\rho)\ge0$ for $X=\frac{1}{d+1}\left(\frac{d}{n}\right)\Big[\tilde{X}+\tr(\tilde{X})I\Big]$; and if the RHS is nonnegative for all real $x \in \text{col}(P)$, $\tr(X^2\rho)\ge0$ for all Hermitian $X$. We conclude that probability-assignments $P(E_i|\rho)$ are valid if and only if
\begin{align}
\label{the_uncertainty_principle}
&\forall x\in \text{col}(P):\\
&\text{Var}[X]_\rho\ge \frac{d}{d+2}\Big(\langle X^2\rangle_\mu -2\langle X\rangle_\mu\langle X\rangle_\rho\Big)-\langle X\rangle_\rho^2,\nonumber
\end{align}
where  $\text{Var}[X]_\rho=\sum_i x_i^2 P(E_i|\rho) - \left(\sum_i x_i P(E_i|\rho)\right)^2$. In this way, the shape of quantum state-space can be understood in terms of a variance-based uncertainty principle: valid probability-assignments on reference outcomes cannot be too sharp lest they violate a lower-bound on the variance for any observable in $\text{col}(P)$.  We note that a related inequality was derived recently in \cite{chen2024nonstabilizernessenhancesthriftyshadow} in the context of bounding the variance of expectation values in shadow estimation.

\section{The Jordan product}

We can shed further light on the special role that 3-designs play in encoding the shape of quantum state-space by considering their relationship to the \emph{Jordan algebra} of observables. To see, this, let us return to the inequality in Eq.\! \ref{vn_var}, which we may reshape  into $\forall x: \sum_i x_i [\mathcal{L}_\rho]_{ij} x_j \ge0$, where
\begin{align} 
	\label{L}
	[\mathcal{L}_\rho]_{ij}=\sum_{kl} \Re[\tr(E_i\sigma_j\sigma_k)]\Phi_{kl}P(E_l|\rho).
\end{align}
	 From this we conclude that probability-assignments $P(E_i|\rho)$ are valid if and only if $\mathcal{L}_\rho$ is positive-semidefinite. In particular, for an unbiased 3-design reference device, Eq.\! \ref{L} simplifies to
\begin{align}
	&[\mathcal{L}_\rho]_{ij}\\
	&=\frac{1}{2}\Bigg[(d+1)(d+2)\left(\frac{n}{d}\right)\sum_mP(E_m|\sigma_i)P(E_m|\sigma_j)P(E_m|\rho)\nonumber\\
	&- P(E_i|\sigma_j) -P(E_i|\rho)-P(E_j|\rho)-\frac{d}{n}\Bigg],\nonumber
\end{align}
which has the virtue of depending only upon reference device probabilities, and which provides a straightforward way to check whether the variance-bound is satisfied.

But there is another interpretation of the operator $\mathcal{L}_\rho$. Recall that under the Jordan product $A\circ B=\frac{1}{2}(AB+BA)$, Hermitian matrices over $\mathbb{C}$ form a \emph{Euclidean Jordan algebra} \cite{Faraut1994,stacey2023quantumtheorysymmetrybroken, wilce2017royalroadquantumtheory}. A Jordan algebra is a nonassociative algebra which satisfies the commutative law and the Jordan identity,
\begin{align}
A \circ B &= B \circ A \nonumber\\
A^2 \circ (B \circ A) &= A \circ (B \circ A^2) \nonumber.
\end{align}
If we define $L_A$ to be the linear operator which takes the Jordan product with $A$, that is, $L_A(B)=A\circ B$, the Jordan identity is equivalent to $\big[L_A, L_{A^2}\big]=0$. A \emph{Euclidean} Jordan algebra enjoys the additional property that there exists an inner-product on the underlying vector space $\mathcal{V}$ such that $\forall A, B, C \in \mathcal{V}: \langle L_A (B), C\rangle=\langle B, L_A(C)\rangle$.

 By introducing an informationally-complete reference device, we identify quantum states with probability distributions. We may then ask: how can we represent the Jordan product in terms of probabilities? If $|\rho \circ \tau)=L_\rho|\tau)$,  using the resolution of the identity $\B{S}\Phi \B{E}=I$, we have
\begin{align}
\B{E}L_\rho|\tau)=\B{E}L_\rho\B{S}\Phi \B{E}|\tau)=\mathcal{L}_\rho\Phi P(E|\tau),
\end{align}
where $\mathcal{L}_\rho=\B{E}L_\rho\B{S}$, whose matrix-elements are
\begin{align}
\label{Lelements}
[\mathcal{L}_\rho]_{ij} &= \tr(E_i L_\rho(\sigma_j))=\frac{1}{2}\big(\tr(E_i\rho \sigma_j) + \tr(E_i\sigma_j\rho)\big)\nonumber\\
&=	\sum_{kl} \Re[\tr(E_i\sigma_j\sigma_k)]\Phi_{kl}P(E_l|\rho).
\end{align}
Indeed, this is precisely the matrix we developed earlier, whose positive-semidefiniteness diagnoses the validity of probability-assignments $P(E_i|\rho)$. 

Significantly, Eq.\! \ref{Lelements} reveals that the three-index tensor $ \Re[\tr(E_i\sigma_j\sigma_k)]$ encodes the structure-coefficients for the Jordan product on $d\times d$ Hermitian matrices over $\mathbb{C}$, and thus fully defines it by its action on the reference states and effects. Since pure-states are idempotents of the Jordan algebra  with trace 1, the structure-coefficients implicitly determine the geometry of the state-space. At the same time, we have shown that the components of this tensor can be extracted from the joint probability distribution $P(E_i,E_j,E_k|\mathcal{M}_3)$.  Finally, taking our reference device to be an unbiased 3-design means that  $P(E_i,E_j,E_k|\mathcal{M}_3)=\frac{1}{n}\sum_m P(E_i|\sigma_m)P(E_j|\sigma_m)P(E_k|\sigma_m)$ so that $P(E_i|\sigma_j)$ \emph{alone} is sufficient to characterize the Jordan product, and through this algebraic structure, the entire state-space. 

\section{Conclusion}

We have thus shown that the shape of quantum state-space can be understood in terms of an uncertainty principle which constrains the probabilities one ought to assign to the outcomes of a 3-design reference device. Compatibility with this uncertainty principle can be diagnosed through the positive-semidefiniteness of a particular operator $\mathcal{L}_\rho$ constructed from reference probabilities. We have also provided a set of scalar constraints that pick out pure-state probability-distributions, which can alternatively be summarized by a single vector constraint. Conceptually, these constraints be understood as entropic uncertainty principles, and operationally they relate to the agreement-probability on multiple copies of the reference device.

 Crucially, each term that appears in our equations is grounded in a probability-assignment, and even better, these probabilities refer to the behavior of a single measure-and-reprepare reference device. The possibility of achieving this rests on the delicate interplay between unitary symmetry and the Jordan algebra of observables.  The algebraic structure of quantum theory implies that the 3rd moment of quantum state-space determines them all, and so does a reference measurement furnished by a 3-design. This further vindicates the centrality of 3-designs already suggested by their optimality in classical shadow estimation tasks.

 Our result holds particular significance for the QBist research program in the foundations of quantum mechanics. QBism argues that quantum theory is best understood not as a description of physical reality, but rather as a set of normative guidelines for gambling on the consequences of one's actions in a world undergoing ceaseless creation \cite{fuchs2019qbismquantumtheoryheros, fuchs2023qbismnext, debrota2024quantumdynamicshappenspaper}. Consequently, QBist ``reconstructions'' of quantum mechanics proceed \cite{DeBrota_2021, Appleby_2017} by motivating the constraints on probability-assignments implied by quantum theory in the same spirit in which de Finetti derived the usual rules of probability theory by contemplating what constraints a gambler ought to place on their different bets in order to prevent a sure loss.

The simplicity of our result is very much in the spirit of \cite{stacey2023quantumtheorysymmetrybroken}, which suggests that the constraints implied by quantum theory are in some sense the ``most symmetrical'' compatible with what the author calls ``vitality,'' i.e.\ the non-existence of a hidden-variable model. Indeed, our result shows that these constraints may be understood as a fundamental expression of complementarity. That said, our derivation presumes a prior knowledge of traditional quantum theory: the significance of our work here is that it exposes the structure that must be aimed for in any future reconstructive effort. 

 For instance, previous efforts at QBist reconstruction  \cite{Appleby_2017} took as their starting place the generalization of the image of quantum state-space within the probability simplex induced by a symmetric informationally-complete (SIC) reference measurement. SIC measurements have a host of virtues: the corresponding states form a simplex in quantum state-space whose vertices are pure states; the conjecture of their existence in any Hilbert space dimension has lead to a fruitful and unexpected interplay between physics and algebraic number theory \cite{Appleby_2017}. SIC states, however, only form 2-designs, and thus the Jordan structure-coefficients cannot be extracted directly from probability-assignments assigned to a single reference device. Breaking this barrier is the central innovation of the present work. 
 
More specifically, the SIC-based reconstructive effort began from the observation that the 2-design condition implies  the inner product between any two probability vectors must lie between certain upper and lower bounds. Inspired by this, \cite{Appleby_2017} defined a \emph{qplex} to be set of probability vectors which mutually satisfy these bounds, to which no more elements can be added without inconsistency, and which contains a simplex of pure states (corresponding to a SIC). The goal of the program was to motivate these bounds on independent grounds, situate quantum theory in the vaster landscape of qplexes, and provide a principle by which quantum theory could be identified within this landscape. This approach is consonant with developments in quantum foundations over the last 25 years where the study of so-called \emph{generalized probabilistic theories} \cite{M_ller_2021, barnum2013postclassicalprobabilitytheory} has played a central role in providing new perspectives on quantum theory.

 The authors of \cite{Appleby_2017} demonstrated that any qplex whose symmetry group is a stochastic subgroup of the orthogonal group isomorphic to the unitary group must correspond to quantum theory, and vice versa: moreover, the existence of such a subgroup is equivalent to the existence of a particular SIC. At the same time, the authors left open the possibility of there being a simpler principle which could pick out quantum theory among the qplexes. The present work shows that demanding the fundamental reference measurement to have the properties of a 3-design, rather than a  2-design, means that a single finite set of probability distributions is sufficient to characterize the theory, which is much more tractable, analytically and computationally, as well as more conceptually satisfying.
 
 Just as a qplex generalizes the representation of quantum mechanics according to a SIC, it is natural to consider analogous generalizations of a 3-design representation, what we might call \emph{3-qplexes}. For instance, one could explore the landscape of all state-spaces defined by an uncertainty principle as in Eq.\! \ref{the_uncertainty_principle}, without at first restricting $P(E_i|\sigma_j)$, the probabilities that characterize the reference device, to correspond to an actual quantum 3-design. Which properties of quantum theory are preserved in such generalized theories, and which fall by the wayside? Just as identifying quantum theory among the qplexes led to an alternative characterization of SICs themselves, picking out quantum theory among the 3-qplexes would lead to an alternative characterization of 3-designs. A new approach to identifying and constructing 3-designs, in particular, of self-testing them, would have significant practical application in quantum computing and beyond. 

\section*{Acknowledgments}
The author would like to thank Christopher Fuchs, Blake Stacey, Sachin Gupta, Gianluca Cuffaro, Arthur Parzygnat, and Marcus Appleby for many helpful discussions.

This work was supported in part by National Science Foundation Grant 2210495 and in part by Grant 62424 from the John Templeton Foundation. The opinions expressed in this publication are those of the author and do not necessarily reflect the views of the John Templeton Foundation.

\bibliography{three}

\appendix 

\section{$\Phi$ for an unbiased 2-design}
\label{2-design-Phi}

For an unbiased 2-design $\{\sigma_i\}_{i=1}^n$, $\frac{1}{n}\sum_i \sigma_i^{\otimes 2}=\frac{1}{d(d+1)}(I \otimes I + \mathcal{S})$. Taking the partial transpose yields
\begin{align}
	\frac{1}{n}\sum_i \sigma_i \otimes \sigma_i^T &=\frac{1}{d(d+1)}\left(I \otimes I + |I)(I|\right).
\end{align}
For a pure-state $\sigma$, $|\sigma)(\sigma|=\sigma \otimes \sigma^T$ where $|\sigma)=\text{vec}(\sigma)$. Letting $E_i=\frac{d}{n}\sigma_i$, we arrive at the resolution of the identity $I = 	(d+1)\sum_i |\sigma_i)(E_i| -|I)(I|$. Comparing this to $\B{S}\Phi \B{E}=I=\sum_{ij}\Phi_{ij}|\sigma_i)(E_j|$, it follows that we may take $\Phi=(d+1)I-\frac{d}{n}J$, where $J$ is the matrix of all 1's.

\section{Scalar conditions and overcompleteness}
\label{scalar-overcompleteness}
Let $\rho=\sum_{ij}\Phi_{ij}P(E_j|\rho)\sigma_i$. Then $\tr(\rho^2)=\frac{n}{d}\sum_{ij}\Phi_{ij}P(E_j|\rho)\sum_{kl}P(E_i|\sigma_k)\Phi_{kl}P(E_l|\rho)$. Decomposing $P(E|\rho)=x+y$, for $x\in \text{col}(P)$ and $y \in \text{col}(P)^\perp$,  and using the form of $\Phi$ for an unbiased 2-design, we find $\tr(\rho^2)=(d+1)\left(\frac{n}{d}\right)\sum_ix_i^2 - 1$. Thus if and only if $\sum_i x_i^2=\left(\frac{d}{n}\right)\frac{2}{d+1}$ does $\tr(\rho^2)=1$. The result is analogous for $\tr(\rho^3)=1$, using  the expression for $\Re \big[\tr(E_i \sigma_j \sigma_k)\big]$  derived in section \ref{idempotents}.

\section{The real-part of the triple-products}
\label{big-real-part}
\begin{align}
&\Re \big[\tr(E_i \sigma_j \sigma_k)\big]\\
&=	\frac{1}{2}\Bigg[(d+1)(d+2)\left(\frac{n}{d}\right)\sum_mP(E_i|\sigma_m)P(E_j|\sigma_m)P(E_k|\sigma_m)\nonumber\\
&-  P(E_j|\sigma_k)-P(E_i|\sigma_j)-P(E_i|\sigma_k)-\frac{d}{n}\Bigg]\nonumber
\end{align}

\section{Simplifying the variance}
\label{variance-overcompleteness}
Let $x$ be an assignment of real numerical values to the outcomes of an unbiased 3-design reference device. This is equivalent to the assignment of a self-adjoint operator $X$ since $\langle X\rangle_\rho=\sum_i x_i P(E_i|\rho)=\tr\left(\sum_i x_iE_i\rho\right)=\tr(X\rho)=\langle X\rangle_\rho^{(\text{vn})}$, although in an overcomplete representation different choices of $x$ will yield the same operator $X$. The variance with respect to a standard von Neumann measurement of $X$ is $\text{Var}[X]^{\text{(vn)}}_\rho = \tr(X^2\rho)-\tr(X\rho)^2$, where $\tr(X^2\rho)= \left(\frac{d}{n}\right)\sum_{ijkl} x_ix_j\Re[\tr(E_i\sigma_j\sigma_k)]\Phi_{kl}P(E_l|\rho)$. Substituting in the expression for $\Re[\tr(E_i\sigma_j\sigma_k)]$ yields
\begin{small}
	\begin{align}
	\label{ugly_variance}
&\tr(X^2\rho)\\
&=\frac{1}{2}\Bigg[(d+1)(d+2)\sum_{k} \left(\sum_i P(E_k|\sigma_i)x_i\right)^2P(E_k|\rho)\nonumber\\
&-\left(\frac{d}{n}\right)\sum_{ij}x_iP(E_i|\sigma_j)x_j-2d\langle X\rangle_\mu\langle X\rangle_\rho-d^2\langle X\rangle_\mu^2\Bigg]\nonumber,
\end{align}	
\end{small}
\!\!Let us assume that $x \in \text{col}(P)$ so that $x_i=\tr(E_i \tilde{X})$ for some Hermitian $\tilde{X}$. By the 2-design property,
\begin{small}
\begin{align}
\sum_jP(E_i|\sigma_j)x_j&=
	\sum_j	P(E_i|\sigma_j)\tr(E_j\tilde{X}) \\
	&=d\tr\left(\frac{1}{n}\sum_j \sigma_j^{\otimes 2} (E_i \otimes \tilde{X}) \right)\\
	&=\frac{1}{d+1}\tr\big((I + \mathcal{S})(E_i \otimes \tilde{X})\big)\\
	&=\frac{1}{d+1}\left(\frac{d}{n}\tr(\tilde{X})+\tr(E_i\tilde{X})\right),
\end{align}
\end{small}
where $\tr(\tilde{X})=\sum_i \tr(E_i\tilde{X})=\sum_i x_i=n\langle X\rangle_\mu$, so that 
\begin{align}
	\sum_j P(E_i|\sigma_j)x_j&=\frac{d\langle X\rangle_\mu +x_i}{d+1}\\
	\left(\sum_j P(E_i|\sigma_j)x_j \right)^2&=\frac{d^2\langle X\rangle_\mu^2+2d\langle X\rangle_\mu x_i+x_i^2}{(d+1)^2}\\
	\sum_{ij}x_iP(E_i|\sigma_j)x_j&=\frac{n\big(d\langle X\rangle_\mu^2+\langle X^2\rangle_\mu\big)}{d+1}.
\end{align}
 Substituting these expressions into Eq.\ \ref{ugly_variance} yields Eq.\ \ref{simplified_var}.

\end{document}